\documentclass[reprint,superscriptaddress,showpacs,amsmath,amssymb,aps,pra,12pt]{revtex4-2}
\usepackage{times}
\usepackage{graphicx}
\usepackage{dsfont}
\usepackage{amsthm}
\usepackage{mathrsfs}
\usepackage{subfigure}
\usepackage[ruled]{algorithm2e}
\usepackage{hyperref}

\def\tr{\mathrm{Tr}}
\hypersetup{colorlinks=true, citecolor=blue, urlcolor=blue, linkcolor=blue}

\begin{document}
\title{Sharing EPR steering between sequential pairs of observers}
\author{Qiao-Qiao Lv}
\affiliation{School of Mathematical Sciences, Capital Normal University, Beijing 100048, China}
\author{Jin-Min Liang}
\affiliation{School of Mathematical Sciences, Capital Normal University, Beijing 100048, China}
\author{Zhi-Xi Wang}
\email{wangzhx@cnu.edu.cn}
\affiliation{School of Mathematical Sciences, Capital Normal University, Beijing 100048, China}
\author{Shao-Ming Fei}
\email{feishm@cnu.edu.cn}
\affiliation{School of Mathematical Sciences, Capital Normal University, Beijing 100048, China}

\begin{abstract}
The recycling of quantum correlations has attracted widespread attention both theoretically and experimentally. Previous works show that bilateral sharing of nonlocality is impossible under mild measurement strategy and 2-qubit entangled state can be used to witness entanglement arbitrary many times by sequential and independent pairs of observers. However, less is known about the bilateral sharing of EPR steering yet. Here, we aim at investigating the EPR steering sharing between sequential pairs of observers. We show that an unbounded number of sequential Alice-Bob pairs can share the EPR steering as long as the initially shared state is an entangled
two-qubit pure state. The claim is also true for particular class of mixed entangled states.
\end{abstract}

\maketitle
\onecolumngrid
\baselineskip18pt

\section{Introduction}
An essential difference between quantum and classical worlds is the existence of quantum correlations like quantum steering \cite{Nielsen2000Quantum,UCNG2020,xiang2022}. In order to response the Einstein-Podolsky-Rosen paradox \cite{EPR1935}, Schr\"odinger initially proposed the concept of quantum steering in 1936 \cite{S1936}. Until 2007, Wiseman \emph{et al.} gave a strict definition of quantum steering by using the local-hidden-state (LHS) model \cite{WJD2007,JWD2007}. Unlike quantum entanglement and Bell nonlocality, one characteristic of steering is the asymmetry between observers \cite{HR2013,BVQB2014,MNBV2021}. It has recently aroused more and more interest, both form a theoretical and experimental \cite{SJWP20210,WWBWP2016,ZKC2020} perspective. People have been committed to the research of quantum steering because of the useful characteristics of steering and its potential applications in quantum information, such as one-side device-independent quantum key distribution \cite{BCWSW2012}, random-number generation \cite{SC2018,JSWP2022}, metrology \cite{YFG2021}, and so on.

The recycling of entangled systems has attracted widespread attention both theoretically \cite{CJAHWA2017,BMSS2018,HDHS2018,NDGRM2020,SGGP2015} and experimentally \cite{HZHL2018,FCTS2020,FRTL2020,FPAT2021}. The issue of the shareability of quantum correlations among multiple sequential observers was first introduced by Silva \emph{et al.} \cite{SGGP2015}. They investigated the shareability of nonlocality based on optimal weak measurements and the Clauser-Horne-Shimony-Holt (CHSH) inequality \cite{CHSH}. Multiple independent observers can share the nonlocality with half of an entangled pair by using optimal weak measurements.

In the seminal work \cite{BC2020} the authors studied fundamental limits on the nonlocality shareability and demonstrated that arbitrarily many independent observers can share the nonlocality of a maximally entangled qubit state with a single observer, where a single Alice can share the Bell-nonlocality with a sequence of Bobs who make local measurements each. Recently, in Ref. \cite{Ztg2021} the authors studied the nonlocal shareability based on an arbitrary dimensional bipartite entangled, but not necessary maximally entangled pure state. However, whether all nonlocal quantum states can generate such nonlocal correlations for arbitrary many independent observers remains a vital open problem. The above both works \cite{BC2020,Ztg2021} consider only the one side shareability of nonlocality in which a single Alice shares nonlocality with a sequence of Bobs. Two sides shareability has been also considered in Refs. \cite{Csm2021,Csm2022}. Limitations on sharing nonlocality between sequential pairs of observers is investigated via the trade off relation between the strengths and maximum reversibilities of qubit measurements. For mixed entangled bipartite states, in particular, if Alice and Bob all perform a sequence of positive operator valued measures (POVM) measurements on each side, the Bell inequality is not always violated \cite{Ba2002}.

Recently, studies on $n$-recyclable locally entangled states related to entanglement detection have attracted widespread attention \cite{CMU2022,CMUS2022}. It has been proved that there exist entangled states that do not violate the three-setting Cavalcanti, Jones, Wiseman and Reid (CJWR) linear steering inequality, whose entanglement can be detected by one Alice and arbitrary many Bobs by using suitable entanglement witness operators and measurement settings. Moreover, an arbitrary number of pairs of sequential and independent observers can witness entanglement \cite{EW,JP,EWO} by operating locally on 2-qubit entangled states, even the entanglement of the initial shared state tends to zero. More recently, recycled detection and sequential detection of genuine multipartite entanglement are investigated \cite{CMU202205,CMU202208}, which shows that it is possible to sequentially detect genuine multipartite entanglement arbitrarily many times for an arbitrarily large number of parties.

As a kind of important quantum correlations the Einstein-Podolsky-Rosen (EPR) steering has recently attracted significant interest in quantum information theory \cite{UCNG2020}. A single pair of entanglement quantum state generates a long sequence of steering correlations under sequential independent measurements on one or each qubit of the quantum state. This phenomenon is referred to as the steering shareability. EPR steering shareability with respect to weak measurement has been discussed in \cite{HQY2019,Paul2020,CHPL2020,ZHLGZ2022}. For 2-qubit system, based on unbiased inputs, a single Alice can steer two Bobs when the double violation of the Clauser-Horne-Shimony-Holt-like inequality is realized. Meanwhile, three Bobs can be steered by a single Alice when the triple violation of the three-setting linear steering inequality is observed \cite{RY2021}. Moreover, the authors conjecture that, at most $n$ Bobs can steer Alice's system based on an $n$-steering linear steering inequality \cite{SDMM2018}. For isotropic entangled states with local dimension $d$, the number of Bobs that can steer Alice is $N_{Bob}\sim \frac{d}{\log d}$ \cite{HDHS2018}. The unilateral sharing of quantum steering has been considered in \cite{HDHS2018,RY2021,SDMM2018}. Steering is intrinsically a directional form of correlation. The problem of manipulating the quantum steering direction with sequential unsharp measurements has recently been investigated in \cite{HQFXG2022}.

Although sharing quantum nonlocality and quantum entanglement with independent observables or sequential pairs of observers have been already investigated, less is known about sharing EPR steering between sequential pairs of observers. In this work, we consider the two sides shareability of steering. We investigate the EPR steering sharing between sequential pairs of observers. We consider whether an initial entangled state can be recycled to generate EPR steering between multiple independent observers on each side. We show that a 2-qubit entangled state can generate steering arbitrary many times by sequential and independent pairs of observers, when the initial pair of observers share any pure entangled state or a class of mixed entangled states.

\section{Preliminary}
Without loss of generality, an arbitrary 2-qubit state $\rho_{AB}$ shared between Alice and Bob can be written as,
\begin{align}\label{state1}
\rho_{AB}=\frac{1}{4}(\mathds{I}_2\otimes\mathds{I}_2+\vec{a}\cdot\vec{\sigma}\otimes\mathds{I}_2+\mathds{I}_2\otimes\vec{b}\cdot\vec{\sigma}
+\sum^3_{i,j=1}t_{ij}\sigma_i\otimes\sigma_j),
\end{align}
where $\mathds{I}_2$ denotes the identity operator with size $2\times2$, $\vec{a}=(a_1, a_2, a_3)$ and $\vec{b}=(b_1, b_2, b_3)$ are real vectors, $\vec{\sigma}=(\sigma_1, \sigma_2,\sigma_3)$ with $\sigma_i$ $(i=1,2,3)$ the standard Pauli matrices. We denote $T=(t_{ij})$ the correlation matrix with entries given by $t_{ij}=\tr(\rho_{AB}\sigma_i\otimes\sigma_j)$.
Under local unitary transformations, $\rho_{AB}$ can be written as \cite{LSL2008}
\begin{align}\label{state2}
\rho^{'}_{AB}=\frac{1}{4}(\mathds{I}_2\otimes\mathds{I}_2
+\vec{a^{'}}\cdot\vec{\sigma}\otimes\mathds{I}_2
+\mathds{I}_2\otimes\vec{b^{'}}\cdot\vec{\sigma}
+\sum^3_{i=1}t_i\sigma_i\otimes\sigma_i),
\end{align}
where the correlation matrix of $\rho^{'}_{AB}$ is diagonal, $T^{'}=diag(t_1, t_2, t_3)$ with $t_i=\tr(\rho^{'}_{AB}\sigma_i\otimes\sigma_i)$.

We briefly recall the concepts of steering and the three-setting CJWR linear steering inequality. For a 2-qubit state $\rho_{AB}$, we say that Alice can steer Bob if the probability correlation $P(a,b|A,B, \rho_{AB})$ can not be described by any local hidden state (LHS) model,
\begin{align}
P(a,b|A,B, \rho_{AB})=\sum_{\lambda}p_{\lambda}P(a|A, \lambda)P_Q(b|B, \rho_{\lambda}),
\end{align}
where $P(a,b|A,B, \rho_{AB})=\tr(A_a\otimes B_b\rho_{AB})$ is the probability of getting outcomes $a$ and $b$ when measurements $A$ and $B$ are performed on $\rho_{AB}$ by Alice and Bob, respectively, $A_a$ and $B_b$ are the corresponding measurement operators, $\lambda$ is the hidden variable, $\rho_{\lambda}$ is the state that Alice sends with probability $p_{\lambda}, \sum_{\lambda}p_{\lambda}=1$, $P(a|A, \lambda)$ is the conditioned probability that Alice obtains measurement outcome $a$ with respect to $\lambda$, $P_Q(b|B, \rho_{\lambda})$ is the quantum probability of getting measurement outcome $b$ obtained by measuring $B$ on the local hidden state $\rho_{\lambda}$.

There are many ways to determine whether a given bipartite quantum state is steerable from Alice to Bob or not. Here, we are interested in the CJWR type of linear steering inequality \cite{CJWR}. The following series of steering inequalities have been derived for the case that both parties are allowed to perform $n$ dichotomic measurements on their respective subsystems \cite{CJWR},
\begin{align}
F_n(\rho_{AB}, \nu)=\frac{1}{\sqrt{n}}|\sum^n_{l=1}\langle A_l\otimes B_l\rangle|\leq1,
\end{align}
where $A_l=\hat{a}_l\cdot\vec{\sigma}$, $B_l=\hat{b}_l\cdot\vec{\sigma}$, $\hat{a}_l\in\textrm{R}^3$ are unit vectors, $\hat{b}_l\in\textrm{R}^3$ are orthonormal vectors, $\nu=\{\hat{a}_1, \hat{a}_2, \ldots, \hat{a}_n, \hat{b}_1, \hat{b}_2, \ldots, \hat{b}_n\}$ denotes the set of measurement directions, $\langle A_l\otimes B_l\rangle=\tr(A_l\otimes B_l\rho_{AB})$. Particularly, the three-setting linear steering inequality can be expressed as \cite{CA},
\begin{align}\label{steeringinequality}
F_3(\rho_{AB}, \nu)=\frac{1}{\sqrt{3}}|\sum^3_{l=1}\langle A_l\otimes B_l\rangle|\leq1.
\end{align}
Violation of the inequality (\ref{steeringinequality}) implies that the state $\rho_{AB}$ is steerable from Alice to Bob. The steering observable is given by the maximum violation
\begin{align}
S(\rho_{AB})=\max_{A_l, B_l}\frac{1}{\sqrt{3}}|\sum^3_{l=1}\langle A_l\otimes B_l\rangle|.
\end{align}
As the states given in (\ref{state1}) and (\ref{state2}) are local unitary equivalent, the steering observable has the canonical form \cite{Paul2020},
\begin{align}
S(\rho_{AB})=S(\rho^{'}_{AB})=\sqrt{\tr(T^tT)}=\sqrt{\tr(T^{'t}T^{'})},
\end{align}
where $T^t$ represents the transpose of the matrix $T$.

We consider the following measurement strategy ${\cal A}_k, {\cal B}_k$, $k=1, 2, \dots, n$, where ${\cal A}_k$ (${\cal B}_k$) takes measurement setting using measurement operators $\{{\cal A}^{(i)}_k, \mathds{I}_2-{\cal A}^{(i)}_k\}$ ($\{{\cal B}^{(i)}_k, \mathds{I}_2-{\cal B}^{(i)}_k\}$) with
\begin{align}\label{MS}
{\cal A}^{(i)}_k=\frac{\mathds{I}_2+\lambda^{(i)}_k \sigma_i}{2},~~~{\cal B}^{(i)}_k=\frac{\mathds{I}_2+\eta^{(i)}_k \sigma_i}{2},~~~
i=1,2,3,
\end{align}
$\lambda^{(i)}_k,\,\eta^{(i)}_k\in(0,1)$ are the sharpness parameters of the corresponding measurements.

\begin{figure}[ht]
\includegraphics[scale=0.5]{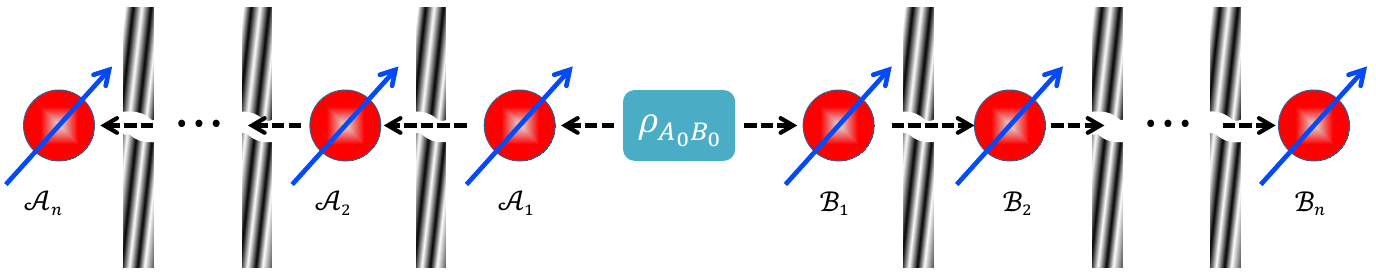}
\caption{Alice and Bob initially share 2-qubit state $\rho_{A_0B_0}$. Alice and Bob preform a series of independent measurements ${\cal A}_k, {\cal B}_k, k=1,2,\dots, n$, respectively.}
\end{figure}

Now consider the following measurement scenario (FIG. 1). The $k$-th Bob ${\cal B}_k$ performs measurement ${\cal B}_k$ and passes his part of the shared state to the $k+1$-th Bob ${\cal B}_{k+1}$. Similarly, the $j$-th Alice ${\cal A}_j$ performs measurement ${\cal A}_j$ and passes her part of the shared state to the ${j+1}$-th Alice ${\cal A}_{j+1}$. Based on the L\"{u}ders rule \cite{L2009}, ${\cal A}_0$ and ${\cal B}_k$ share a state
\begin{align}\label{LB1}
\rho_{A_0B_k}=&\frac{1}{3}\sum^3_{i=1}\Big(\mathds{I}_2\otimes\sqrt{{\cal B}^{(i)}_k}
\rho_{A_0B_{k-1}}\mathds{I}_2\otimes\sqrt{{\cal B}^{(i)}_k}\nonumber\\
&+\mathds{I}_2\otimes\sqrt{\mathds{I}_2-{\cal B}^{(i)}_k}
\rho_{A_0B_{k-1}}\mathds{I}_2\otimes\sqrt{\mathds{I}_2-{\cal B}^{(i)}_k}\Big).
\end{align}
By the chain rules, ${\cal A}_j$ and ${\cal B}_k$ share a state
\begin{align}\label{LB2}
\rho_{A_jB_k}=&\frac{1}{3}\sum^3_{i=1}\Big(\sqrt{{\cal A}^{(i)}_j}
\otimes\mathds{I}_2\rho_{{\cal A}_{j-1}B_k}\sqrt{{\cal A}^{(i)}_j}\otimes\mathds{I}_2\nonumber\\
&+\sqrt{\mathds{I}_2-{\cal A}^{(i)}_j}\otimes\mathds{I}_2
\rho_{A_{j-1}B_k}\sqrt{\mathds{I}_2-{\cal A}^{(i)}_j}\otimes\mathds{I}_2\Big).
\end{align}

\section{Shareability of EPR steering for pure and mixed entangled states}
In this section, we show that two qubits can both be recycled to generate EPR steering between independent observers on each side. Let $\rho_{A_0B_0}$ be the state initially shanred by $A_0$ and $B_0$. We assume that in the measurement strategy (\ref{MS}) adopted by ${\cal A}_k$ and ${\cal B}_k$, $\lambda^{(1)}_k=\lambda^{(2)}_k=\lambda_{k}$ and $\eta^{(1)}_k=\eta^{(2)}_k=\eta_{k}$, $\lambda^{(3)}_k=\eta^{(3)}_k=1$,  where $\lambda_k, \eta_k\in(0,1)$ \cite{CMU2022}.
By using the following relations \cite{BC2020},
\begin{align}
\sqrt{{\cal B}^{(i)}_k}=\sqrt{\frac{\mathds{I}_2+\eta^{(i)}_k \sigma_i}{2}}
=\frac{\sqrt{1+\eta^{(i)}_k}+\sqrt{1-\eta^{(i)}_k}}{2\sqrt{2}}\mathds{I}_2
+\frac{\sqrt{1+\eta^{(i)}_k}-\sqrt{1-\eta^{(i)}_k}}{2\sqrt{2}}\sigma_i
\end{align}
and
\begin{align}
\sqrt{\mathds{I}_2-{\cal B}^{(i)}_k}=\sqrt{\frac{\mathds{I}_2-\eta^{(i)}_k \sigma_i}{2}}
=\frac{\sqrt{1+\eta^{(i)}_k}+\sqrt{1-\eta^{(i)}_k}}{2\sqrt{2}}\mathds{I}_2
-\frac{\sqrt{1+\eta^{(i)}_k}-\sqrt{1-\eta^{(i)}_k}}{2\sqrt{2}}\sigma_i,
\end{align}
we obtain the state shared by $A_0$ and $B_k$,
\begin{align}\label{stateBk}
\rho_{A_0B_k}=&\frac{1}{3}\left(\frac{3+2\sqrt{1-\eta^2_{k-1}}}{2}\rho_{A_0B_{k-1}}
+\frac{1}{2}\mathds{I}_2\otimes\sigma_3\rho_{A_0B_{k-1}}\mathds{I}_2
\otimes\sigma_3\right.\\\nonumber
&\left.+\frac{1-\sqrt{1-\eta^2_{k-1}}}{2}(\mathds{I}_2
\otimes\sigma_1\rho_{A_0B_{k-1}}\mathds{I}_2\otimes\sigma_1
+\mathds{I}_2\otimes\sigma_2\rho_{A_0B_{k-1}}\mathds{I}_2\otimes\sigma_2)\right).
\end{align}
We have the following conclusion.

\newtheorem{theorem}{Theorem}
\begin{theorem}
The bilateral EPR steering sharing can be realized for the initial 2-qubit maximally entangled state $|\psi\rangle_{A_0B_0}=\frac{|00\rangle+|11\rangle}{\sqrt{2}}$.
\end{theorem}

\begin{proof}
After Alice and Bob independently perform measurements once, we get
\begin{align}\label{state11}
\rho_{A_1B_1}=&\frac{1}{18}\left((5+s_1 \left(2 t_1+1\right)+t_1)|00\rangle\langle00|+\left(s_1+1\right) \left(t_1+1\right)|00\rangle\langle11|\right.\\\nonumber
&+(4-s_1 \left(2 t_1+1\right)-t_1)|01\rangle\langle01|+(4-s_1 \left(2 t_1+1\right)-t_1)|10\rangle\langle10|\\\nonumber
&\left.+\left(s_1+1\right) \left(t_1+1\right)|11\rangle\langle00|+(s_1 \left(5+2 t_1+1\right)+t_1)|11\rangle\langle11|\right),
\end{align}
where $s_1=\sqrt{1-\lambda^2_1}$ and $t_1=\sqrt{1-\eta^2_1}$. By directly computation, we get the correlation matrix of $\rho_{A_1B_1}$,
\begin{align}
T_{\rho_{A_1B_1}}=\frac{1}{9}\left(\begin{array}{ccc}
(s_1+1)(t_1+1)&0&0\\
0&-(s_1+1)(t_1+1)&0\\
0&0&(2s_1+1)(2t_1+1)
\end{array}\right).
\end{align}
Then the steering observable is given by
\begin{align}
S(\rho_{A_1B_1})=\tr(T^t_{\rho_{A_1B_1}}T_{\rho_{A_1B_1}})=\frac{1}{81} \left(2 \left(s_1+1\right){}^2 \left(t_1+1\right){}^2+\left(2 s_1+1\right){}^2 \left(2 t_1+1\right){}^2\right).
\end{align}
It can be numerically shown that there exist $s_1$ and $t_1$ such that $S(\rho_{A_1B_1})>1$. Namely, there exist measurement strategies to realize the bilateral EPR steering sharing.
\end{proof}

Theorem 1 shows that the EPR steering sharing can be realized for one time bilateral measurements. Such bilateral measurements are not unique. To study the $n$ times EPR steering sharing, we consider particular measurements with $\lambda_k=\eta_k$, $k=1,2$.

\begin{theorem}
For arbitrary 2-qubit entangled pure quantum state $|\varphi\rangle_{A_0B_0}$ with Schmidt decomposition $|\varphi\rangle_{A_0B_0}=\sqrt{\alpha}|00\rangle+\sqrt{1-\alpha}|11\rangle$, the $n$-recycle EPR steering sharing can be realized for $\alpha\in(0,\frac{1}{2}]$.
\end{theorem}

\begin{proof} Starting from the initial state $|\varphi\rangle_{A_0B_0}$ and using the measurement strategy (\ref{MS}) with $\lambda_k=\eta_k$, $k=1,2$, we obtain the correlation matrix $T_{\xi_{A_kB_k}}$ of $\xi_{A_kB_k}$,
\begin{align}
T_{\xi_{A_kB_k}}=\frac{1}{9^k}\left(\begin{array}{ccc}
2\prod^k_{i=1}\sqrt{ \alpha(1-\alpha) } \left(s_i+1\right)^2&0&0\\
0&-2\prod^k_{i=1}\sqrt{ \alpha(1-\alpha) } \left(s_i+1\right)^2&0\\
0&0&\prod^k_{i=1}(2s_i+1)^2
\end{array}\right),
\end{align}
where $|\varphi\rangle_{A_kB_k}$ is the state after Alice and Bob make $k$ times measurements respectively, $s_i=\sqrt{1-\lambda^2_i}$. By straightforward calculation we get
\begin{align}
S(\xi_{A_kB_k})=\frac{1}{81^k} \left(8\alpha(1-\alpha)\prod^k_{i=1}\left(s_i+1\right){}^4+\prod^k_{i=1}\left(2 s_i+1\right){}^4\right),
\end{align}
which is monotonically increasing with $\alpha$.

Assume $s_i=1-f(k)$ such that $f(k)\rightarrow 0$ when $k\rightarrow\infty$ for $k\in \textrm{N}$, where $k$ is number of the measurements that Alice and Bob perform. We have
\begin{align}
S(\xi_{A_kB_k})=\frac{1}{81^k} \left(8\alpha(1-\alpha)\left(2-f(k)\right)^{4k}+\left(3-2f(k)\right)^{4k}\right).
\end{align}
We observe that $S(\xi_{A_kB_k})\rightarrow 1$ when $k\rightarrow\infty$ for $k\in \textrm{N}$. Moreover, for any $\alpha$ we can always choose $f(1)$ such that $S(\xi_{A_1B_1})>1$. Hence $S(\xi_{A_kB_k})>1$ for all $k\in \textrm{N}$. Therefore, an unbounded number of sequential Alice-Bob pairs can share the EPR steering as long as the initial pair is  entangled.
\end{proof}

{\sf Remark:}

(1) For $\alpha=\frac{1}{2}$, $|\varphi\rangle_{A_0B_0}$ is a maximally entangled state. We can choose $f(k)=10^{-k}$.  Based on the relarion between $f(k)$ and $\lambda_k$, we can get the measurement strategies. Then $S(\xi_{A_kB_k})>1$ for all $k\in\textrm{N}$.

(2) For smaller $\alpha$, for instance $\alpha=0.1$, we can choose $f(k)=100^{-k}$.  Simalar to (1), we can obtain the corresponding measurement strategies. We still have $S(\xi_{A_kB_k})>1$ for all $k\in\textrm{N}$.

(3) When $\alpha$ tends to zero, in order to ensure that $\rho_{A_kB_k}$ is steerable, $s_i$ needs to tend to 1. Then the selection of measurement strategies will decrease.

Hence, the closer $\alpha$ tends to zero, the less measurement strategies we can choose, as less entanglement implies less amount of quantum resource.

Next we consider the shareability of EPR steering for mixed entangled states. We first consider the following class of entangled states,
\begin{align}\label{MS1}
\gamma_{A_0B_0}=m_1|\varphi\rangle_{A_0B_0}\langle\varphi|+m_2|00\rangle\langle00|+m_3|11\rangle\langle11|,
\end{align}
where $|\varphi\rangle_{A_0B_0}=\sqrt{\alpha}|00\rangle+\sqrt{1-\alpha}|11\rangle$, $m_1>0$, $m_2, m_3\geq0$ and $m_1+m_2+m_3=1$, $\alpha\in(0,\frac{1}{2}]$. 
When $m_1>0$, $m_2, m_3\geq0$ and $m_1+m_2+m_3=1$, $\alpha\in(0,\frac{1}{2}]$, $\gamma_{A_0B_0}$ is entangled.
Moreover, $\gamma_{A_0B_0}$ is steerable for $8\alpha(1-\alpha)m^2_1+1>1$.
By direct calculation we obtain the correlation matrix of $\gamma_{A_0B_0}$,
\begin{align}
T_{\gamma_{A_0B_0}}=\frac{\prod^k_{i=1}(2s_i+1)^2}{9^k}\left(\begin{array}{ccc}
2\sqrt{ \alpha(1-\alpha) }m_1 &0&0\\
0&-2\sqrt{ \alpha(1-\alpha) }m_1&0\\
0&0&1
\end{array}\right),
\end{align}
where $s_i=\sqrt{1-\lambda^2_i}$. Then
\begin{align}
S(\gamma_{A_kB_k})=\frac{1}{81^k} \left(8\alpha(1-\alpha)m^2_1\prod^k_{i=1}\left(s_i+1\right){}^4+\prod^k_{i=1}\left(2 s_i+1\right){}^4\right).
\end{align}
Considering $s_i=1-f(k)$ such that $f(k)\rightarrow 0$ when $k\rightarrow\infty$ for $k\in \textrm{N}$, we obtain
\begin{align}
S(\gamma_{A_kB_k})=\frac{1}{81^k} \left(8\alpha(1-\alpha)m^2_1\left(2-f(k)\right){}^{4k}+\left(3-2f(k)\right){}^{4k}\right).
\end{align}
$S(\gamma_{A_kB_k})\rightarrow 1$ when $k\rightarrow\infty$ for all $k\in \textrm{N}$ and $S(\gamma_{A_1B_1})>1$ for suitable $f(k)$. Namely, $S(\gamma_{A_kB_k})>1$ for all $k\in \textrm{N}$. Therefore, there are unbounded sequential Alice-Bob pairs to share the EPR steering when the initial pair is a mixed steerable state.

For weakly entangled mixed states, we consider the following 2-qubit state,
\begin{align}\label{state}
\tilde{\gamma}_{A_0B_0}=\frac{1}{4}(\mathds{I}_2\otimes\mathds{I}_2-\cos\theta\sigma_1\otimes\sigma_1-\alpha\sin\theta\sigma_2\otimes\sigma_2
-\alpha\sin\theta\sigma_3\otimes\sigma_3),
\end{align}
where $\frac{\sqrt{2}}{2}<\alpha\leq 1$ and $\theta\in(0,\frac{\pi}{4}]$. 
When $\frac{\sqrt{2}}{2}<\alpha\leq 1$ and $\theta\in(0,\frac{\pi}{4}]$, $\tilde{\gamma}_{A_0B_0}$ is entangled.
Moreover, when $\cos^2\theta+\alpha^2\sin^2\theta>1$, $\tilde{\gamma}_{A_0B_0}$ is steerable. After $k$ times measurements, the correlation matrix of $\tilde{\gamma}_{A_kB_k}$ is given by
\begin{align}
T_{\tilde{\gamma}_{A_kB_k}}=-\frac{1}{9^k}\left(\begin{array}{ccc}
\cos\theta\prod^k_{i=1}\left(s_i+1\right)^2&0&0\\
0&\alpha\sin\theta\prod^k_{i=1} \left(s_i+1\right)^2&0\\
0&0&\alpha\sin\theta\prod^k_{i=1}(2s_i+1)^2
\end{array}\right).
\end{align}
Then
\begin{align}
S(\tilde{\gamma}_{A_kB_k})=\frac{1}{81^k} \left(\cos^2\theta\prod^k_{i=1}\left(s_i+1\right){}^4+\alpha^2\sin^2\theta\prod^k_{i=1}\left(s_i+1\right){}^4
+\alpha^2\sin^2\theta\prod^k_{i=1}\left(2 s_i+1\right){}^4\right).
\end{align}
Considering $s_i=1-f(k)$ such that $f(k)\rightarrow 0$ when $k\rightarrow\infty$ for $k\in \textrm{N}$, we obtain
\begin{align}
S(\tilde{\gamma}_{A_kB_k})=\frac{1}{81^k} \left(\cos^2\theta\left(2-f(k)\right){}^{4k}+\alpha^2\sin^2\theta\left(2-f(k)\right){}^{4k}
+\alpha^2\sin^2\theta\left(3-2f(k)\right){}^{4k}\right).
\end{align}
We find that $S(\tilde{\gamma}_{A_kB_k})\rightarrow \alpha^2\sin^2\theta$ as $k\rightarrow\infty$. However, as $\alpha^2\sin^2\theta<1$ for $\frac{\sqrt{2}}{2}<\alpha\leq 1$, $\theta\in(0,\frac{\pi}{4}]$, the EPR steering sharing for an unbounded number of sequential Alice-Bob pairs is impossible for this initial mixed entangled state.

Finally, we consider 2-qubit Werner states,
\begin{align}
	\gamma^W_{A_0B_0}=m|\varphi^{'}\rangle_{A_0B_0}\langle\varphi^{'}|+(1-m)\frac{\mathds{I}_2}{2}\otimes\frac{\mathds{I}_2}{2},
\end{align}
where $|\varphi^{'}\rangle_{A_0B_0}=\cos\theta|00\rangle+\sin\theta|11\rangle$, $\theta\in[0,\frac{\pi}{4}]$, $0\leq m<1$. $\gamma^W_{A_0B_0}$ is entangled for $m\geq\frac{1}{1+2\sin(2\theta)}$. After $k$-th measurement, the correlation matrix of $\gamma^W_{A_kB_k}$ is given by
\begin{align}
	T_{\gamma^W_{A_kB_k}}=-\frac{1}{9^k}\left(\begin{array}{ccc}
		m\sin(2\theta)\prod^k_{i=1}\left(s_i+1\right)^2&0&0\\
		0&-m\sin(2\theta)\prod^k_{i=1}\left(s_i+1\right)^2&0\\
		0&0&m\prod^k_{i=1}(2s_i+1)^2
	\end{array}\right).
\end{align}
Then
\begin{align}
	S(\gamma^W_{A_kB_k})=\frac{m^2}{81^k} \left(2\sin^2(2\theta)\prod^k_{i=1}(s_i+1){}^4+\prod^k_{i=1}\left(2s_i+1\right){}^4\right).
\end{align}
Considering $s_i=1-f(k)$ such that $f(k)\rightarrow 0$ when $k\rightarrow\infty$ for $k\in \textrm{N}$, we obtain
\begin{align}
	S(\gamma^W_{A_kB_k})=\frac{m^2}{81^k} \left(2\sin^2(2\theta)\prod^k_{i=1}(f(k))+\prod^k_{i=1}\left(3-2f(k)\right){}^4\right).
\end{align}
We have $S(\gamma^W_{A_kB_k})\rightarrow m^2<1$ as $k\rightarrow\infty$. The EPR steering sharing for an unbounded number of sequential Alice-Bob pairs is impossible when the initial mixed entangled state is a Werner state.

Based on the analysis on the three types of mixed states above, we see that there are unbounded sequential Alice-Bob pairs to share the EPR steering when the initially shared states are given by (\ref{MS1}). When the initially shared states are given by (\ref{state}) or the Werner states, there are bounded sequential Alice-Bob pairs to share the EPR steering.

\section{Conclusion}
We have investigated whether an initial entangled state can both be recycled to generate EPR steering between multiple independent observers on each side. Different from previous study \cite{HDHS2018,RY2021}, they only consider the scenario that a single Alice steers multiple Bobs. In this wok,  we consider the EPR steering sharing between sequential pairs of observers. Enriching the results on recycling of the quantum steering, we have found that a 2-qubit entangled state can be used to generate such steerability arbitrary many times by sequential and independent pairs of observers. The conclusion is also true when the initial pair of observers share a pure entangled state or a particular class of mixed entangled states, but not all entangled mixed states according to our measurement strategy. Whether there are other methods to increase the number of bilateral sharing for all mixed entangled state is still an open question.
Our approach may also highlight researches on recycling other quantum correlations.

\begin{acknowledgements}
This work is supported by the National Natural Science Foundation of China (NSFC) under Grants 12075159 and 12171044; Beijing Natural Science Foundation (Grant No. Z190005); the Academician Innovation Platform of Hainan Province.
\end{acknowledgements}

\end{document}